\documentclass[12pt,a4paper]{article}
\usepackage{verbatim}
\usepackage[utf8]{inputenc} 
\usepackage[english]{babel} 
\usepackage{anysize} 
\usepackage{indentfirst}
\usepackage{amsmath}
\usepackage{amsthm} 
\usepackage{amsfonts}
\usepackage{amssymb}
\usepackage{graphicx}
\usepackage{enumitem}

\graphicspath{{figuras/}}
\usepackage{subfig} 
\usepackage{xcolor}

\usepackage{fancyhdr}
\usepackage{calligra}

\newtheorem{theorem}{Theorem}[section]
\newtheorem{proposition}[theorem]{Proposition}

\newtheorem{definition}[theorem]{Definition}
\newtheorem{example}[theorem]{Example}

\newcommand{\etal}{\textit{et al.}\ }

\title{Abstract homogeneity of interval-valued functions}

\author{Ana Shirley Monteiro \and Regivan Santiago \and Radko Mesiar \and Marisol Gomez \and Martin Papco \and Mikel Ferrero-Jaurrieta \and Humberto Bustince}

\date{}

\begin{document}

\maketitle


\section{Introduction}

In this work we develop abstract homogeneity (introduced by Santiago \etal (see \cite{SantiAbsHomo2021},  \cite{Santiago2021FHom})) to the interval context by considering admissible orders. We will show some properties and how it appears in fuzzy logic framework.

\section{Abstract homogeneity over IV-functions}

The following definition naturally extends abstract homogeneity of real-valued functions to interval-valued ones.

\begin{definition} [Abstract homogeneous IV-function]
\label{def_intabshomfunction}
Let be $G: \mathbb{I}([0,\!1])^2 \to \mathbb{I}([0,\!1])$  and $\varPhi: \mathbb{I}([0,\!1]) \to \mathbb{I}([0,\!1])$ be an order isomorphism. An IV-function $F: \mathbb{I}([0,\!1])^n \to \mathbb{I}([0,\!1])$ is said to be \emph{abstract homogeneous with respect to $G$ and $\varPhi$}, or \emph{$(G,\varPhi)$-homogeneous}, if
\begin{equation*}
F\big(G(\Lambda,X_1), \ldots,G(\Lambda,X_n)\big) = G\big(\varPhi(\Lambda),F(X_1, \ldots, X_n)\big)
\end{equation*}
for all $\Lambda, X_1, \ldots, X_n\in\mathbb{I}([0,\!1])$. If $\varPhi$ is the identity, then $F$ is simple said to be \emph{$G$-homogeneous}.
\end{definition}

\begin{proposition}\label{prop_intabshom}
Let $\pi_2: \mathbb{I}([0,\!1])^2 \to \mathbb{I}([0,\!1])$ be  given by $\pi_2(X_1,X_2)=X_2$ for all $(X_1,X_2)\in\mathbb{I}([0,\!1])^2$.Then every IV-function $F: \mathbb{I}([0,\!1])^n \to \mathbb{I}([0,\!1])$ is $\pi_2$-homogeneous.
\end{proposition}

\begin{proof}
Let $\pi_2: \mathbb{I}([0,\!1])^2 \to \mathbb{I}([0,\!1])$ be given by $\pi_2(X_1,X_2)=X_2$ for all $(X_1,X_2)\in\mathbb{I}([0,\!1])^2$, and let $F: \mathbb{I}([0,\!1])^n \to \mathbb{I}([0,\!1])$ be an IV-function. Then for all $\Lambda,X_1,\ldots, X_n\in\mathbb{I}([0,\!1])$,
\begin{equation*}
F\big(\pi_2(\Lambda,X_1),\ldots,\pi_2(\Lambda,X_n)\big)
=F(X_1,\ldots, X_n)
=\pi_2\big(\Lambda,F(X_1, \ldots, X_n)\big)
\end{equation*}
Thus $F$ is $\pi_2$-homogeneous.
\end{proof}

\begin{theorem}
\label{theo_homidempfunction}
Let $G: \mathbb{I}([0,\!1])^2 \to \mathbb{I}([0,\!1])$ be a function. If $F: \mathbb{I}([0,\!1])^n \to \mathbb{I}([0,\!1])$ is  $G$-homogeneous and there exists $A \in \mathbb{I}([0,\!1])$ such that $F(A, \ldots, A) = A$ and $G_{A}: \mathbb{I}([0,\!1]) \to \mathbb{I}([0,\!1])$, given by $G_{A}(X) = G(X,A)$,  is a bijection, then $F$ is idempotent.
\end{theorem}

\begin{proof}
\begin{flushleft}
Let $X \in \mathbb{I}([0,\!1])$. Since $G_{A}$ is a bijection on $\mathbb{I}([0,\!1])$, $X = G_{A}(Y)$ for some $Y \in \mathbb{I}([0,\!1])$. Thus, from the $G$-homogeneity of $F$, it follows:
\begin{align*}
F(X, \ldots, X) & = F(G_{A}(Y), \ldots, G_{A}(Y)) = F(G(Y,A), \ldots, G(Y,A)) = G(Y, F(A, \ldots, A))\\
& = G(Y, A) = G_{A}(Y)= X.	
\end{align*}
\end{flushleft}
\end{proof}

\section{\label{sec:4}Homogeneity in IV-fuzzy logic}

\begin{proposition}\label{prop_intnegfunction}
Let $F: \mathbb{I}([0,\!1])^n \to \mathbb{I}([0,\!1])$ be an IV-function and let $P: \mathbb{I}([0,\!1])^2 \to \mathbb{I}([0,\!1])$ be the IV-product function, i.e. it is given by $P(X,Y) = X \cdot Y$. If $F$ is $P$-homogeneous then the IV-function $F_{N_S}: \mathbb{I}([0,\!1])^n \to \mathbb{I}([0,\!1])$ given by $F_{N_S}\big(X_1, \ldots, X_n\big) = N_S\big(F\big(N_S(X_1), \ldots, N_S(X_n)\big)\big)$ is homogeneous with respect to $P_{N_S}(X,Y)
= N_S\big(P\big(N_S(X),N_S(Y)\big)\big)$.
\end{proposition}
\begin{proof} Due to the previous definition,
\begin{equation*}
\underline{F_{N_S}(X_1, \ldots, X_n)} = 1 - \overline{F(\boldsymbol{1} - X_1, \ldots, \boldsymbol{1} - X_n)}
\end{equation*}
and
\begin{equation*}
\overline{F_{N_S}(X_1, \ldots, X_n)} = 1 - \underline{F(\boldsymbol{1} - X_1,\ldots, \boldsymbol{1} - X_n)}.
\end{equation*}
Regarding $P_{N_S}$,
\begin{align*}
\underline{P_{N_S}(X,Y)}
&= 1 - \overline{P(\boldsymbol{1} - X, \boldsymbol{1} - Y)}
= 1 - \overline{\boldsymbol{1} - X} \cdot  \overline{\boldsymbol{1} - Y}
= 1-(1-\underline{X})\cdot(1-\underline{Y})\\
&=\underline{X} + \underline{Y} - \underline{X} \underline{Y}= \underline{X} + \left(1 - \underline{X} \right)  \underline{Y}.
\end{align*}
Similarly, $\overline{P_{N_S}(X,Y)} = \overline{X} + \overline{Y} - \overline{X} \overline{Y} =  \overline{X} + \left(1 - \overline{X} \right)  \overline{Y}$.

\medskip

Thus, for all $\Lambda, X_1, \ldots, X_n\in \mathbb{I}([0,\!1])$,
\begin{align*}
&\underline{F_{N_S}(P_{N_S}(\Lambda, X_1), \ldots, P_{N_S}(\Lambda, X_n))}\\
&=1 - \overline{F(\boldsymbol{1} - P_{N_S}(\Lambda, X_1), \ldots, \boldsymbol{1} - P_{N_S}(\Lambda, X_n))}\\
&=1 - \overline{F(\boldsymbol{1} - [\underline{\Lambda} + (1 - \underline{\Lambda}) \underline{X_1}, \overline{\Lambda} + (1 - \overline{\Lambda}) \overline{X_1}], \ldots, \boldsymbol{1} - [\underline{\Lambda} + (1 - \underline{\Lambda}) \underline{X_n}, \overline{\Lambda} + (1 - \overline{\Lambda}) \overline{x_n}])}\\
&=1 - \overline{F\left((\boldsymbol{1}-\Lambda)(\boldsymbol{1}-X_1), \ldots,  (\boldsymbol{1}-\Lambda)(\boldsymbol{1}-X_n) \right)}\\
&=1 - \overline{F\left(P(\boldsymbol{1}-\Lambda,\boldsymbol{1}-X_1), \ldots, P(\boldsymbol{1}-\Lambda,\boldsymbol{1}-X_n) \right)}\\
&=1 - \overline{P\left(\boldsymbol{1}-\Lambda, F(\boldsymbol{1}-X_1, \ldots,\boldsymbol{1}-X_n) \right)}\quad\text{[because of $P$-homogeneity of $F$]}\\
&=1 - \overline{(\boldsymbol{1}-\Lambda) \cdot F(\boldsymbol{1}-X_1, \ldots,\boldsymbol{1}-X_n)}\\
&=1 - \overline{\boldsymbol{1}-\Lambda} \cdot \overline{F(\boldsymbol{1}-X_1, \ldots,\boldsymbol{1}-X_n)}\\
&=1 - \left(1 - \underline{\Lambda} \right) \cdot \overline{F(\boldsymbol{1}-X_1, \ldots,\boldsymbol{1}-X_n)}
\end{align*}

Analogously,
\begin{equation*}
\overline{F_{N_S}(P_{N_S}(\Lambda, X_1), \ldots, P_{N_S}(\Lambda, X_n))} = 1 - \left(1 - \overline{\Lambda} \right) \cdot \underline{F(\boldsymbol{1}-X_1, \ldots,\boldsymbol{1}-X_n)}
\end{equation*}

On the other hand,
\begin{align*}
\underline{P_{N_S}(\Lambda, F_{N_S}(X_1, \ldots, X_n))}
&= \underline{\Lambda} + \left(1 - \underline{\Lambda} \right)  \cdot \underline{F_{N_S}(X_1, \ldots, X_n)}\\
&=\underline{\Lambda} + \left(1 - \underline{\Lambda} \right)  \cdot \left(1 - \overline{F(\boldsymbol{1}-X_1, \ldots, \boldsymbol{1}-X_n)}\right)\\
&=1  - \left(1 - \underline{\Lambda} \right)  \cdot \overline{F(\boldsymbol{1}-X_1, \ldots, \boldsymbol{1}-X_n)}.
\end{align*}

Therefore $F_{N_S}$ is $P_{N_S}$-homogeneous.
\end{proof}

\begin{example}
\label{example_minmaxprodhom}
IV-functions $\min$ and $\max$ are $P$-homogeneous. Indeed, for all $\Lambda,X,Y\in\mathbb{I}([0,\!1])$,  $\min(P(\Lambda, X), P(\Lambda,Y)) = \min(\Lambda \cdot X, \Lambda \cdot Y) = \Lambda \cdot \min(X,Y) = P(\Lambda, \min(X,Y))$. Analogously, $ \max(P(\Lambda,X), P(\Lambda,Y)) = \Lambda \cdot \max(X,Y) = P(\Lambda, \max(X,Y)).$
\end{example}



\end{document}